\documentclass{IEEEtran}
\usepackage{cite}
\usepackage{url}
\usepackage{amsmath,amssymb,amsfonts}
\usepackage{multirow}
\usepackage{amsthm}

\usepackage{graphicx}

\usepackage{textcomp}
\usepackage{epsfig}
\usepackage{caption}
\usepackage{graphicx}
\usepackage{epstopdf}
\usepackage{amsmath,amssymb}
\usepackage[framemethod=TikZ]{mdframed}

\newmdenv[
  topline=true, bottomline=true,
  leftline=false, rightline=false,
  linewidth=0.6pt,
  innerleftmargin=0pt, innerrightmargin=0pt,
  innertopmargin=4pt, innerbottommargin=4pt
]{AlgRuleBox}

\usepackage{epstopdf}
\usepackage[ruled]{algorithm}
\usepackage{algorithmic}
\usepackage{enumerate}
\usepackage{float}
\usepackage{graphics,graphicx}
\usepackage{subfigure}

\usepackage{array}
\usepackage{mathtools}
\usepackage[justification=centering]{caption}
\allowdisplaybreaks[3]
\usepackage{url}
\DeclareGraphicsExtensions{.png}
\graphicspath{{./Pics/}}
\usepackage{color}
\usepackage{xspace}

\def\tt{^{\rm T}}

\usepackage{hyperref}
\hypersetup{
    colorlinks=true,
    linkcolor=black,
    filecolor=gray,      
    urlcolor=blue,
    citecolor=blue,
}

\def\BibTeX{{\rm B\kern-.05em{\sc i\kern-.025em b}\kern-.08em
		T\kern-.1667em\lower.7ex\hbox{E}\kern-.125emX}}

\newtheorem{lemma}{Lemma}
\newtheorem{theorem}{Theorem}

\newtheorem{remark}{Remark}
\newtheorem{assumption}{Assumption}

\newtheorem{corollary}{Corollary}

\begin{document}
	
	\title{\LARGE \bf  Privacy-Preserving Dynamic Average Consensus by Masking Reference Signals}

\author{Mihitha Maithripala, Zongli Lin

\thanks{Mihitha Maithripala and Zongli Lin are with the Charles L. Brown Department of Electrical and
Computer Engineering, University of Virginia, Charlottesville, VA 22904, U.S.A (e-mail:
wpg8hm@virginia.edu; zl5y@virginia.edu. Corresponding author:
Zongli Lin.)

}
}

\maketitle	
 
\begin{abstract}
In multi-agent systems, dynamic average consensus (DAC) is a decentralized estimation strategy in which a set of agents tracks the average of time-varying reference signals. Because DAC requires exchanging state information with neighbors, attackers may gain access to these states and infer private information. In this paper, we develop a privacy-preserving method that protects each agent’s reference signal from external eavesdroppers and honest-but-curious agents while achieving the same convergence accuracy and convergence rate as conventional DAC. Our approach masks the reference signals by having each agent draw a random real number for each neighbor, exchanges that number over an encrypted channel at the initialization, and computes a masking value to form a masked reference. Then the agents run the conventional DAC algorithm using the masked references. Convergence and privacy analyses show that the proposed algorithm matches the convergence properties of conventional DAC while preserving the privacy of the reference signals. Numerical simulations validate the effectiveness of the proposed privacy-preserving DAC algorithm.

\end{abstract}



\section{Introduction}\label{sec:introduction}
Average consensus has attracted significant attention in multi-agent systems. They have been applied in numerous domains, including distributed load sharing~\cite{meng2016distributed}, distributed economic dispatch~\cite{chen2022privacy,chen2024quantized}, sensor fusion~\cite{aragues2012distributed,olfati2005consensus}, state-of-charge balancing in networked battery energy storage systems~\cite{meng2021distributed}, and multi-robot formation control~\cite{zhang2022privacy}. Average consensus algorithms are commonly classified as static~\cite{olfati2007consensus,ren2010distributed} or dynamic~\cite{bai2010robust,chen2012distributed,chen2014distributed}, depending on whether the goal is to calculate the average of the agent's initial states or to track the average of time-varying reference signals. Recently, there has been growing interest in DAC algorithms for applications, such as distributed control for optimal current sharing and voltage regulation in DC microgrids~\cite{aguirre2025dynamic} and distributed control of networked battery energy storage systems (BESSs)~\cite{maithripala2025privacy}.

An early study on DAC, presented in \cite{spanos2005dynamic}, proposes an algorithm that reaches steady state tracking of the average of time varying reference signals, given predefined initial conditions. Two estimation schemes for DAC, a proportional algorithm and a proportional integral algorithm, are proposed in \cite{freeman2006stability} to yield bounded steady state error. The proportional integral method is further developed in \cite{bai2010robust} to achieve zero steady state error for a class of time varying inputs. DAC algorithms with discontinuous update laws were proposed in~\cite{chen2012distributed} to achieve finite time convergence. Nonlinear DAC schemes were introduced in~\cite{nosrati2012dynamic}. Robust DAC algorithms with discontinuous dynamics resilient to initialization errors are presented in~\cite{george2019robust}. Recent work \cite{xiao2023dynamic} develops linear, exosystem-based DAC schemes that achieve zero steady-state error for signals with known frequency components and are robust to interruptions of the network connectivity.

All mentioned DAC algorithms require to exchange their state values with its neighbors. This direct information exchange remains vulnerable to adversaries or eavesdroppers. If attackers gain access to those state values, they can design targeted strategies to observe their private states like local reference signals and then can disrupt the operation of other agents. In state of charge balancing for a networked battery energy storage system (BESS)~\cite{meng2021distributed}, each battery’s state, representing the energy available at that unit, serves as the reference signal in the consensus scheme. These values should remain private because disclosure to an external party or an insider can destabilize the system or degrade performance. In energy trading applications that use DAC, such disclosure can enable manipulation of bidding strategies, prices, or traded energy volumes, leading to unfair outcomes or instability across the network.

Several approaches have been introduced to mitigate these threats. However, most privacy preserving schemes focus on the static average consensus. Differential privacy \cite{he2020differential,Nozari2015DifferentiallyPA,gao2018differentially} is a commonly used method that preserves private information by injecting noise to the shared data. Because of this added noise, it generally cannot guarantee exact convergence to the true average. In~\cite{mo2016privacy,zhang2022privacy}, it is shown that by carefully designed perturbed signals, accurate convergence can be achieved. A cryptography-based scheme is introduced in~\cite{ruan2019secure}, where exchanged information is encrypted to protect against eavesdroppers. However, these schemes incur substantial computational and communication overhead as encryption and decryption are resource intensive. A state decomposition method for static dynamic average consensus with low computational overhead is proposed in~\cite{wang2019privacy}.
The privacy-preserving schemes considered so far focus on static average consensus. In~\cite{kia2015dynamic}, a privacy-preserving DAC algorithm is proposed that protects each agent’s reference signal against both external and internal attackers. But the privacy guarantee holds only when the reference signals are dynamic over an initial time window. Moreover, the state decomposition method of~\cite{wang2019privacy} has been extended to the DAC setting in~\cite{zhang2022privacy}.

Although the state decomposition based DAC in~\cite{zhang2022privacy} offers strong privacy against external eavesdroppers, the modification of the graph Laplacian degrades the convergence rate relative to that of conventional DAC. Moreover, that scheme protects the privacy of the reference signal only against external eavesdroppers. These limitations motivate the design of a privacy-preserving DAC that (i) matches conventional DAC in both convergence rate and steady-state tracking error, (ii) safeguards each agent’s reference against external eavesdroppers and honest-but-curious internal agents, and (iii) maintains low computational complexity.

In this study, we develop a new privacy-preserving DAC algorithm that uses masked reference signals. During mask generation at the initialization, each agent draws an independent random real number for each neighbor and shares it via an encrypted channel. Each agent then computes a masking value from the shared numbers, obfuscates its reference signal by adding this mask, and finally performs the DAC algorithm using the masked reference. Here, the masking value is computed such that the average of all masking signals equals zero. Therefore, the average of the masked reference signals is equal to the average of the reference signals without the mask. We present a convergence analysis showing that the proposed algorithm converges to the same neighborhood as conventional DAC and attains the same convergence rate, thereby outperforming state decomposition methods that alter the dynamics of the conventional DAC algorithm and degrade convergence. We also establish privacy guarantees against external eavesdroppers by proving that they cannot recover the reference signals of any agent with the available information set. Furthermore, we prove that an honest-but-curious agent cannot infer a target neighbor’s reference signal, provided the target has at least one legitimate neighbor. Finally, simulations validate the analysis and demonstrate the effectiveness of the algorithm.

The remainder of the paper is organized as follows. Section~\ref{sec:preliminaries} presents preliminaries and the problem formulation, including the communication network model and the privacy definition for the DAC algorithm. Section~\ref{sec:Main} introduces the proposed privacy-preserving DAC algorithm and provides convergence and privacy analyses against external eavesdroppers and honest-but-curious agents. Section~\ref{sec:Simulation} reports simulation results that validate the algorithm’s convergence and privacy properties. Finally, Section~\ref{sec:Conclusion} concludes the paper.

\section{Preliminaries and Problem Formulation}\label{sec:preliminaries}
\subsection{Communication Network}
 Consider a connected undirected graph $\mathcal{G}(\mathcal{V}, \mathcal{E})$ consisting of $N\geq3$ agents. The node set is $\mathcal{V} = \{v_1, v_2, \ldots, v_N\}$, and the edge set is $\mathcal{E} \subseteq \mathcal{V} \times \mathcal{V}$. For each node $i$, let $\mathcal{N}_i = \{v_j \in \mathcal{V} : (v_i, v_j) \in \mathcal{E}\}$ denote the set of neighbors, with $N_i = |\mathcal{N}_i|$ representing the number of neighbors. The Laplacian matrix associated with the adjacency matrix $A$ is defined as $L = (l_{ij}) \in \mathbb{R}^{N \times N}$, where $l_{ij} = -a_{ij}$ for $i \neq j$ and $l_{ii} = \sum_{j=1, j \neq i}^{N} a_{ij}$. The matrix $L$ is symmetric and positive semidefinite, with eigenvalues ordered as $0 = \lambda_1 \leq \lambda_2 \leq \cdots \leq \lambda_N$. Let $\mathbf{1}_n$ denote the $n$-dimensional all-ones vector, and let $\|x\|_p$ denote the $p$-norm of a vector $x$, where $p \in [1, \infty]$.

We impose the following assumption on the communication graph $\mathcal{G}$.

\begin{assumption} \label{ass: graph}
The communication graph $\mathcal{G}=(\mathcal{V},\mathcal{E})$ is a connected undirected graph with $N\ge 3$ agents.
\end{assumption}

\subsection{ Dynamic Average Consensus Algorithm and Privacy Definition}\label{sec:attack}

This study considers a network of $N \geq 3$ interconnected agents, each with a time-varying internal state $x_i(t) \in \mathbb{R}$, that aim to cooperatively estimate their average through the DAC algorithm in~\cite{spanos2005dynamic},
\begin{equation}
\begin{aligned}
\dot{\hat{x}}_{{\rm a},i}(t) &= \dot{x}_i(t) - \beta \sum_{j=1}^{N} a_{ij}\big(\hat{x}_{{\rm a},i}(t) - \hat{x}_{{\rm a},j}(t)\big),\\
\hat{x}_{{\rm a},i}(0) &= x_i(0) ,\ \ i\in\mathcal{V}.
\end{aligned}
\label{eq:1}
\end{equation}
Here, $\hat{x}_{{\rm a},i}$ denotes the local estimate of the average $x_{\rm a}=\tfrac{1}{N} \sum_{i=1}^{N} x_i(t)$, $\beta > 0$ is a tunable design parameter, and $a_{ij} \in \{0,1\}$ indicates the communication connectivity between agents $i$ and $j$. Although this approach guarantees consensus, it requires each agent to transmit its estimate $\hat{x}_{{\rm a},i}(t)$, which indirectly reveals information about the underlying signal $x_i(t)$ as well as its time derivative $\dot{x}_i(t)$.

The vulnerability of the DAC scheme was highlighted in~\cite{zhang2022privacy}, where an observer-based attack model was introduced. Assuming that $x_i(t), \dot{x}_i(t), \ddot{x}_i(t) \in \mathcal{L}_\infty$, the study demonstrated that adversaries, either within the network or external to it, can extract or reconstruct sensitive information such as the agents’ local time-varying reference signals $x_i(t)$ through observer-based techniques. These findings indicate that the standard DAC algorithm is inherently prone to privacy breaches unless enhanced with privacy-preserving strategies.

Our objective in this paper is to design a novel privacy preserving DAC algorithm that protects the local time-varying reference signals $x_i(t)$. To formalize this goal, we consider two categories of adversaries, (1) internal honest-but-curious agents, who adhere to the update rule. but may attempt to infer the private states of their neighbors from the exchanged messages, and (2) external eavesdroppers, who can intercept all communication between neighboring agents and may exploit this information to reconstruct the sensitive states of all agents.

\section{Main Results}\label{sec:Main}

\subsection{Privacy-Preserving Algorithm Design }
 We first present a privacy-preserving algorithm with reference signal masking to protect the private information $x_i(t)$ of agents while achieving DAC. The algorithm consists of a masking-reference generation phase and a consensus update phase.
 
 During the masking stage, agent $i$ generates mutually independent real random values $\eta_{ij}$ for each neighbor $j \in \mathcal{N}_i$ and sends $\eta_{ij}$ to neighbor $j$ over an encrypted link. 

Then each agent $i$ calculate the mask $m_{i}$ based on following equation,
 \begin{equation}
             m_i = \sum_{j\in\mathcal{N}_i}\big(\eta_{ji}-\eta_{ij}\big),
         \label{eq:2}
 \end{equation}
and calculate the masked reference signal $ x_{{\rm m}, i}(t)$ as
 \begin{equation}
             x_{{\rm m}, i}(t) = x_i(t) + m_i.
         \label{eq:3}
 \end{equation}

In the updating phase, all agents perform the following updating rule,
\begin{equation}
\begin{aligned}
\dot{\hat{x}}_{{\rm a},i}(t) &= \dot{x}_{\rm m,i}(t) - \beta \sum_{j=1}^{N} a_{ij}\big(\hat{x}_{{\rm a},i}(t) - \hat{x}_{{\rm a},j}(t)\big),\\
\hat{x}_{{\rm a},i}(0) &= x_i(0)+m_i ,\ \ i\in\mathcal{V}.
\end{aligned}
\label{eq:4}
\end{equation}

By differentiating \eqref{eq:3} with respect to time $t$, we obtain $\dot{x}_{\mathrm{m},i}(t) = \dot{x}_i(t)$. Therefore, the new updating rule is identical to the \eqref{eq:1}. Algorithm 1 illustrates the masking phase and DAC updating phase.

 \begin{AlgRuleBox}
\noindent\textbf{Algorithm 1.}\quad {Privacy protection via masking the reference signal $x_i(t)$}
\par\smallskip
\noindent\rule{\linewidth}{0.4pt}

\medskip
$\blacktriangleright$ \textbf{Mask synthesis:}
\begin{enumerate}
  \item For each neighbor $j\in\mathcal{N}_i$, agent $i$ draws an independent real random value $\eta_{ij}$.
  \item Agent $i$ sends $\eta_{ij}$ to neighbor $j$ and, in return, collects $\eta_{ji}$ from $j$.
  \item It then forms the net mask
        \[
        m_i = \sum_{j\in\mathcal{N}_i}\big(\eta_{ji}-\eta_{ij}\big).
        \]
  \item Using~\eqref{eq:3}, agent $i$ obfuscates its private reference signal $x_i(t)$  with the mask $m_i$.
\end{enumerate}

$\blacktriangleright$ \textbf{Embedding into the DAC update:}
\begin{enumerate}
  \item Agent $i$ proceeds with the distributed calculation, applying the update rule \eqref{eq:4}.
\end{enumerate}
\end{AlgRuleBox}

Note that the masked reference signal $x_{{\rm m},i}(t)$ is constructed so that the added random offsets do not affect the final tracking accuracy. In particular, the combination of each agent’s locally generated numbers with those received from its neighbors cancels network-wide, eliminating the net effect of the perturbation. We subsequently establish that the designed algorithm leaves the final tracking accuracy unchanged compared to the conventional DAC algorithm in \eqref{eq:1}.

\begin{remark}
When the network topology changes, the masking setup must be reinitialized. Agents should regenerate the random values and repeat their one-time exchange. Any risk of leakage during this brief exchange is mitigated by encryption~\cite{chen2024new}. Since this random-number exchange occurs only once per initialization, the added computational cost is negligible.
\end{remark}

\begin{assumption}\label{ass:encryption}
During the mask generation phase all neighbor-to-neighbor links are encrypted. For each edge $(i,j)\in\mathcal{E}$ the value $\eta_{ij}$ sent by agent $i$ is known only to $i$ and $j$ and the value $\eta_{ji}$ sent by agent $j$ is known only to $j$ and $i$.
\end{assumption}

\begin{lemma}\label{lemma1}
\cite{kia2019tutorial} If each input signal $x_i(t)$ in~\eqref{eq:1} is bounded, there exists a constant $\gamma>0$ such that, for any gain $\beta>0$, the estimate $\hat{x}_{{\rm a},i}(t)$ produced by \eqref{eq:1} over a connected graph converges to a bounded neighborhood of $\tfrac{1}{N}\mathbf{1}_N\tt x(t)$ at an exponential rate, i.e.,
\begin{align*}
\lim_{t \to \infty} \sup \left| \hat{x}_{{\rm a},i}(t) - \tfrac{1}{N}\mathbf 1_N\tt x(t) \right| &\leq \frac{\gamma}{\beta \lambda_2},\quad i\in\mathcal V . 
\end{align*}
where
\[
\gamma = \sup_{\tau \in [t,\infty)}
\left\|\left( I_N - \tfrac{1}{N}\mathbf{1}_N \mathbf{1}_N\tt \right)\dot{x}(\tau)\right\|, 
\]
$x(t) = [x_1(t)\; x_2(t)\,\dotsc\, x_N(t)]\tt$, and $\lambda_2$ is the smallest nonzero eigenvalue of Laplacian matrix $L$.
\end{lemma}

\begin{theorem}
\label{thm:ppdac-exact}
Let $G(\mathcal V,\mathcal E)$ be a connected undirected graph (Assumption~1), and let $A=[a_{ij}]$ be a doubly stochastic weight matrix. 
Each agent $i$ constructs $m_i$ by \eqref{eq:2} and the masked reference $x_{{\rm m},i}(t)=x_i(t)+m_i$ by \eqref{eq:3}, and runs the DAC update rule~\eqref{eq:4}. 
Then the estimate $\hat{x}_{{\rm a},i}(t)$ produced by \eqref{eq:4} converges to a bounded neighborhood of $\tfrac{1}{N}\mathbf{1}_N\tt x(t)$ as in the conventional DAC (Lemma~\ref{lemma1}) and with an appropriate choice of $\beta$, the bounds can be made arbitrarily small.
\end{theorem}

\begin{proof}
Define $\phi=[\eta_{ij}]\in\mathbb{R}^{N\times N}$ with $\eta_{ij}=0$ if $j\notin\mathcal N_i$ or $j=i$. 
Stack $m_i$ into $m = [m_1\; m_2\, \dotsc\, m_N]\tt$. From \eqref{eq:2},
\[
m \;=\; (\phi\tt-\phi)\mathbf 1_N .
\]
Hence
\[
\mathbf 1_N\tt m \;=\; \mathbf 1_N\tt(\phi\tt-\phi)\mathbf 1_N
= \mathbf 1_N\tt\phi\tt\mathbf 1_N - \mathbf 1_N\tt\phi\mathbf 1_N
= 0,
\]
so the masks sum to zero.

Let $x_m(t) = [x_{{\rm m},1}(t)\; x_{{\rm m},2}(t)\, \dotsc\, x_{{\rm m},N}(t)]\tt$. Hence
\[
\mathbf 1_N\tt x_{\rm m}(t)=\mathbf 1_N\tt x(t)+\mathbf 1_N\tt m=\mathbf 1_N\tt x(t).
\]

Differentiating \eqref{eq:3} gives $\dot x_{{\rm m},i}(t)=\dot x_i(t)$ (the masks $m_i$ are time-invariant). 
Stacking \eqref{eq:4} yields
\[
\dot{\hat x}_{\rm a}(t) \;=\; \dot x_{\rm m}(t) - \beta L\,\hat x_{\rm a}(t)
\;=\; \dot x(t) - \beta L\,\hat x_{\rm a}(t),
\]
where $L$ is the graph Laplacian induced by $A$.
This is exactly the conventional DAC updating rule.

We can get the tracking error for the new algorithm as 
\begin{align*}
\lim_{t \to \infty} \sup \left| \hat{x}_{{\rm a},i}(t) - \tfrac{1}{N}\mathbf 1_N\tt x_{\rm m}(t) \right| &\leq \frac{\gamma}{\beta \lambda_2},\quad i\in\mathcal V . 
\end{align*}
where
\[\gamma = \sup_{\tau \in [t,\infty)}
\left\|\left( I_N - \tfrac{1}{N}\mathbf{1}_N \mathbf{1}_N\tt \right)\dot{x}(\tau)\right\|. 
\]

Since $\mathbf{1}_N\tt x_{\rm m}(t)=\mathbf{1}_N\tt x(t)$ for all $t$ and $\dot{x}_{{\rm m},i}(t)=\dot{x}_i(t)$, the above tracking error can be reformulated as 
\begin{align*}
\lim_{t \to \infty} \sup \left| \hat{x}_{{\rm a},i}(t) - \tfrac{1}{N}\mathbf 1_N\tt x(t) \right| &\leq \frac{\gamma}{\beta \lambda_2},\quad  i\in\mathcal V . 
\end{align*}
where
\[\gamma = \sup_{\tau \in [t,\infty)}
\left\|\left( I_N - \tfrac{1}{N}\mathbf{1}_N \mathbf{1}_N\tt \right)\dot{x}(\tau)\right\|. 
\]
Thus, by Lemma~\ref{lemma1}, the estimate $\hat{x}_{{\rm a},i}(t)$ generated by~\eqref{eq:4} converges to the same bounded neighborhood of $\tfrac{1}{N}\mathbf{1}_N\tt x(t)$ as in the conventional DAC. The bounds can be made as small as desired by adjusting the design gain $\beta$.
\end{proof}

\begin{remark}
The convergence rate of consensus algorithms is governed by the  $\lambda_2$ (the second-smallest eigenvalue) of the graph Laplacian $L$~\cite{olfati2007consensus}. The DAC convergence rate is no worse than $\beta\lambda_2$~\cite{kia2019tutorial}. By Theorem~1, the masking of $x_i(t)$ leaves $L$ unchanged. Therefore, the convergence rate of the proposed algorithm is the same as in the conventional DAC setting.

\end{remark}

\subsection{Privacy Analysis }

We analyze the privacy-preserving properties of the proposed algorithm against external eavesdroppers and honest-but-curious adversaries in this subsection.

\subsubsection{Privacy Analysis Against External Eavesdroppers}

Since an external eavesdropper can wiretap all network links, the information available to it are defined as
\[
I_e=\{\, A,\ \beta,\ \hat{x}_{{\rm a},i}^{\alpha}(t)\mid i\in\mathcal V,\ t\ge 0 \,\}.
\]

\begin{theorem}
\label{ext-privacy}
Under Assumptions~\ref{ass: graph} and~\ref{ass:encryption}, the proposed privacy-preserving algorithm prevents external eavesdroppers from uniquely identifying any agent’s reference trajectory $x_i(t)$. 
\end{theorem}

\begin{proof}
Let $x(t)=[x_1(t)\; x_2(t)\, \dotsc\, x_N(t)]\tt$, $m = [m_1\; m_2\, \dotsc\, m_N(t)]\tt$, and $x_{\rm m}(t)=x(t)+m$ in \eqref{eq:3}. By construction of the masks in \eqref{eq:2} and encryption in the first phase, the eavesdropper does not know $m$ and only sees quantities that are functions of $x_{\rm m}(t)$ and the algorithm states driven by $x_{\rm m}(t)$ (e.g., $\hat{x}_{{\rm a},i}^{\alpha}(t)$). Moreover, $\mathbf{1}_N\tt m=0$, so $\frac{1}{N}\mathbf{1}_N\tt x_{\rm m}(t)=\frac{1}{N}\mathbf{1}_N\tt x(t)$.

Consider any vector $s\in\mathbb{R}^N$ with $\mathbf{1}_N\tt s=0$ and define an alternative pair
\[
x'(t)=x(t)+s,\quad m'=m-s.
\]

Then \(x_{\rm m}'(t)=x'(t)+m'=x(t)+m=x_{\rm m}(t)\).
Since \(s\) and \(m\) are time-invariant vectors, we have
\(\dot{m}'=\dot{m}=0\) and \(\dot{x}'(t)=\dot{x}(t)\), hence
\[
\dot{x}_{\rm m}'(t)=\dot{x}'(t)+\dot{m}'=\dot{x}(t)=\dot{x}_{\rm m}(t).
\]
Therefore, for these alternative pairs \((x,m)\) and \((x',m')\), all transmitted signals produced by the update law~\eqref{eq:4} are identical and the estimate $\hat{x}_{{\rm a},i}^{\alpha}(t)$ driven by  \((x',m')\) converges to a neighborhood of the $\tfrac{1}{N}\mathbf{1}_N\tt x(t)$ which is similar to the \((x,m)\) case.
An external eavesdropper observing the transcript \(I_e\) cannot tell the two cases apart.
Consequently, \(x(t)\) is not uniquely recoverable by external eavesdroppers. Any $x'(t)=x(t)+s$ with $\mathbf{1}_N\tt s=0$ looks the same to the observer.

Hence, under the proposed privacy-preserving algorithm, each agent’s trajectory $x_i(t)$ is protected from external eavesdroppers.
\end{proof}

\subsubsection{Privacy Analysis Against Honest–but–Curious Agents}
Let $H$ be the set of honest–but–curious nodes and $L$ the set of legitimate nodes, with $H\cap L=\varnothing$ and $H\cup L=\mathcal V$. 
Honest–but–curious nodes may collude and share what they observe. 
A legitimate node $l\in L$ follow the consensus algorithm as specified and do not try to infer their neighbors’ states. 
Since an honest–but–curious node  receive everything sent by its neighbors, the information available to an honest–but–curious node $h\in H$ is defined as
\[
\mathcal{I}_h=\{\, A,\ \beta,\ \hat{x}_{{\rm a},h}^{\alpha}(t), \ \hat{x}_{{\rm a},s}^{\alpha}(t), \ x_h(t) , \ \dot{x}_h(t), \eta_{hs} , \eta_{sh} \mid \ t\ge 0 \,\}.
\]
where $ s\in\mathcal N_h $. Accordingly, we define the information set observable by the honest–but–curious agents as
\[
\mathcal{I}_H = \{\, \mathcal{I}_h \mid h \in H \,\}.
\]

\begin{theorem}
\label{hbc-privacy}
Under Assumptions~\ref{ass: graph} and~\ref{ass:encryption}, the proposed privacy-preserving algorithm ensures that an honest–but–curious agent learns only the pairwise values on its own incident edges.
Let $\mathcal H\subseteq H$ be any coalition of honest–but–curious agents that may share all observations. 
Fix a target legitimate agent $i\notin\mathcal H$. 
If at least one neighbor of $i$ does not collude (i.e., $\mathcal N_i\setminus\mathcal H\neq\varnothing$) which means that at least one neighbor of $i$ is legitimate, then the true reference trajectory $x_i(t)$ of target legitimate agent $i$ is not uniquely identifiable by honest-but-curious agents. 
\end{theorem}

\begin{proof}
Let $x(t)=[x_1(t)\; x_2(t)\, \dotsc\, x_N(t)]\tt$, masks $m = [m_1\; m_2\, \dotsc\, m_N]\tt$, and masked references $x_{\rm m}(t)=x(t)+m$. 
For each neighbor $j\in\mathcal N_i$, the mask $m_i$ contains the edge terms $(\eta_{ji}-\eta_{ij})$. 
The coalition $\mathcal H$ knows exactly those $\eta_{hi}$ and $\eta_{ih}$  lying on edges incident to agents in $\mathcal H$. 
Thus, for the target $i\notin\mathcal H$, it can split

\[
m_i=\underbrace{\sum_{h\in\mathcal N_i\cap\mathcal H}(\eta_{hi}-\eta_{ih})}_{v_i\ \text{(known to }\mathcal H)}
\;+\;
\underbrace{\sum_{l\in\mathcal N_i\setminus\mathcal H}(\eta_{li}-\eta_{il})}_{u_i\ \text{(unknown to }\mathcal H)}.
\]
Hence, the coalition of honest-but-curious agents can reconstruct only
\[
x_{{\rm m},i}(t)=x_i(t)+m_i=\bigl(x_i(t)+u_i\bigr)+v_i \text{(known constant)},
\]
i.e., it observes the sum $x_i(t)+u_i$ (up to the known constant $v_i$), not $x_i(t)$ and $u_i$ separately.

Fix any $r\in\mathbb R$ and pick a legitimate neighbor
$l\in\mathcal N_i\setminus\mathcal H$.
Define a new execution by
\[
x'(t)=x(t)+r\,e_i-r\,e_l,\quad
m'=m-r\,e_i+r\,e_l.
\]
Here, $e_i$ and $e_l$ denote vectors with a $1$ in position $i$ (respectively, $l$) and $0$ elsewhere.

Realize $m'$ by changing only the unobserved pair on edge $(i,l)$. Set $\eta'_{il}=\eta_{il}+r$ and $\eta'_{l i}=\eta_{l i}$,
and keep \emph{all} $\eta$ values on edges incident to $\mathcal H$ unchanged.

Then
\[
m_i'=m_i-r,\quad m_l'=m_l+r,\quad m_j'=m_j\ (j\neq i,l),
\]
so the known part $v_i$ is unchanged while the unknown part becomes
$u_i' = u_i - r$.
Consequently,
\[
\begin{aligned}
x'_{{\rm m},i}(t)
&= x'_i(t)+u_i'+v_i
= \bigl(x_i(t)+r\bigr)+\bigl(u_i-r\bigr)+v_i \\[2pt]
&= x_i(t)+u_i+v_i
= x_{{\rm m},i}(t).
\end{aligned}
\]

This means that the same masked signal can be generated under $(x,m)$ and $(x',m')$, yet $x'_i(t)=x_i(t)+r$ differs by an arbitrary constant shift. Since the same masked signal is generated, the resulting transmitted signals produced by the updating law~\eqref{eq:4} are identical, and the estimates converge to a neighborhood of $\tfrac{1}{N}\mathbf{1}_N\tt x(t)$. Therefore, the set of honest–but–curious agents observing the information set $I_\mathcal H$ cannot tell the two cases apart.

Since $r$ is arbitrary, $x_i(t)$ is not uniquely identifiable from the honest-but-curious agent’s observations. any decomposition $(x_i(t)+u_i)=(x_i(t)+r)+(u_i-r)$ is observationally equivalent. This means that if at least one neighbor of $i$ is not an honest–but–curious agent (i.e., $\mathcal N_i\setminus\mathcal H\neq\varnothing$), then the true reference trajectory $x_i(t)$ of the legitimate node $i$ is not uniquely identifiable by the honest–but–curious nodes.
\end{proof}

\begin{corollary}
If the target node $i$ has only one neighbor or if all of its neighbors are honest-but-curious agents ($\mathcal N_i\subseteq\mathcal H$), then $u_i=0$ is known and the coalition can recover $x_i(t)$ exactly. Therefore,  privacy fails in these cases.
\end{corollary}

\section{Simulation Results}\label{sec:Simulation}

We validate the proposed privacy-preserving DAC algorithm through the simulation results presented in this section.

Consider a connected, undirected ring network with six agents as shown in Figure.~\ref{fig:fig1}.
\begin{figure}[!t]
\centering
\includegraphics[width=0.4\linewidth]{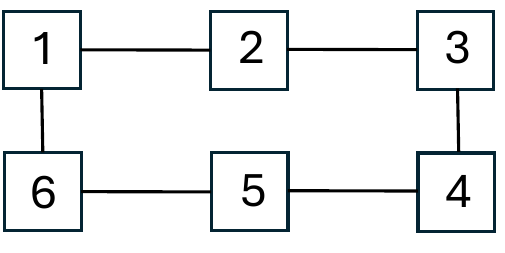}
\caption{The communication topology}
\label{fig:fig1}
\end{figure}
For each agent \(i\in\{1, 2,\dots,6\}\), the reference signals \(x_i(t)\) are
\begingroup\setlength{\jot}{1pt} 
\[
\begin{alignedat}{2}
x_1(t)&=-7\sin(0.5t-\tfrac{2\pi}{3}),\; &\; x_2(t)&=-6.5\sin(0.75t-\tfrac{\pi}{3}),\\[-1pt]
x_3(t)&=-6\sin t,\;                      &\; x_4(t)&=-5.5\cos(1.25t+\tfrac{\pi}{3}),\\[-1pt]
x_5(t)&=-5\cos(1.5t+\tfrac{2\pi}{3}),\;  &\; x_6(t)&=-4.5\cos(1.75t+\pi).
\end{alignedat}
\]
\endgroup
All angles are in radians.

To proceed with Algorithm~1, each agent randomly generates mutually independent real random values $\eta_{ij}$ for its neighbors, collected in the matrix
\[
\phi =
\begin{bmatrix}
0 & 0.20 & 0 & 0 & 0 & -0.40 \\
6.75 & 0 & 1.10 & 0 & 0 & 0 \\
0 & -0.60 & 0 & 0.80 & 0 & 0 \\
0 & 0 & -10.25 & 0 & 0.50 & 0 \\
0 & 0 & 0 & -0.75 & 0 & 1.30 \\
7.90 & 0 & 0 & 0 & -3.20 & 0
\end{bmatrix}.
\]
Note that the \(i\)th row, \(j\)th column entry of matrix gives the random number.

The calculated mask for all agents are
\[
m=\bigl[\,14.85\;-8.25\;-9.35\;9.80\;-3.25\;-3.80\,\bigr]\tt.
\]

We set the initial conditions of the estimators as \( \hat{x}_{{\rm a},i}(0) = x_i(0) +m_i \) and the design parameter is specified as $\beta = 300$.

We apply Algorithm~1 with the specified parameters. Figure~\ref{fig:fig2} illustrates the dynamic evolution of the estimator states \(\hat{x}_{{\rm a},i}(t)\). The results show that each agent’s estimate closely tracks the true average \(x_{\rm a}(t)\), achieving rapid convergence and high accuracy.
\begin{figure}[!t]
\centering
\includegraphics[width=3.6in]{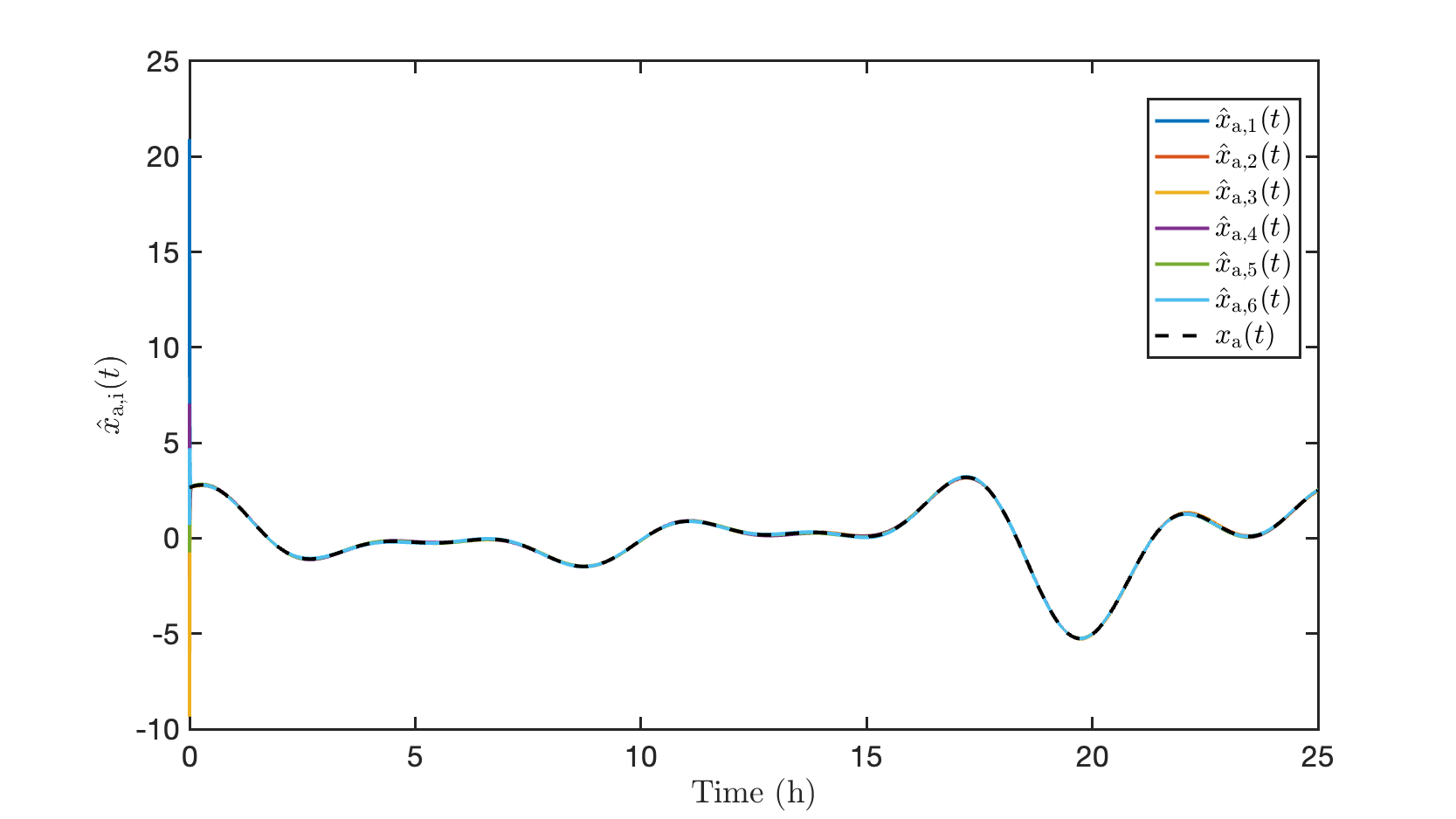}
\caption{The estimator state $\hat{x}_{{\rm a},i}$ convergence.}
\label{fig:fig2}
\end{figure}

We demonstrate that the proposed privacy-preserving DAC algorithm is resilient to external eavesdroppers. By Theorem~\ref{ext-privacy}, each agent’s reference signal \(x_i(t)\) remains protected from any external observer. To empirically validate this result, we implement the attacker observer model in~\cite{zhang2022privacy}, which attempts to reconstruct the true references \(x_i(t)\) from shared messages. 

We simulate a scenario in which an external eavesdropper attempts to infer each agent’s reference signal $x_i(t)$. As shown in Figure.~\ref{fig:fig3}, the attacker recovers only the masked reference \(x_{{\rm m},i}(t)=x_i(t)+m_i\). (Its estimate \(\hat{x}_{\rm e,i}(t)\) of the true reference signal does not match \(x_i(t)\).)
This confirms that, without knowledge of the private mask \(m_i\), an external eavesdropper cannot reconstruct the exact reference signal $x_i(t)$.
\begin{figure}[!t]
\centering
\includegraphics[width=3.6in]{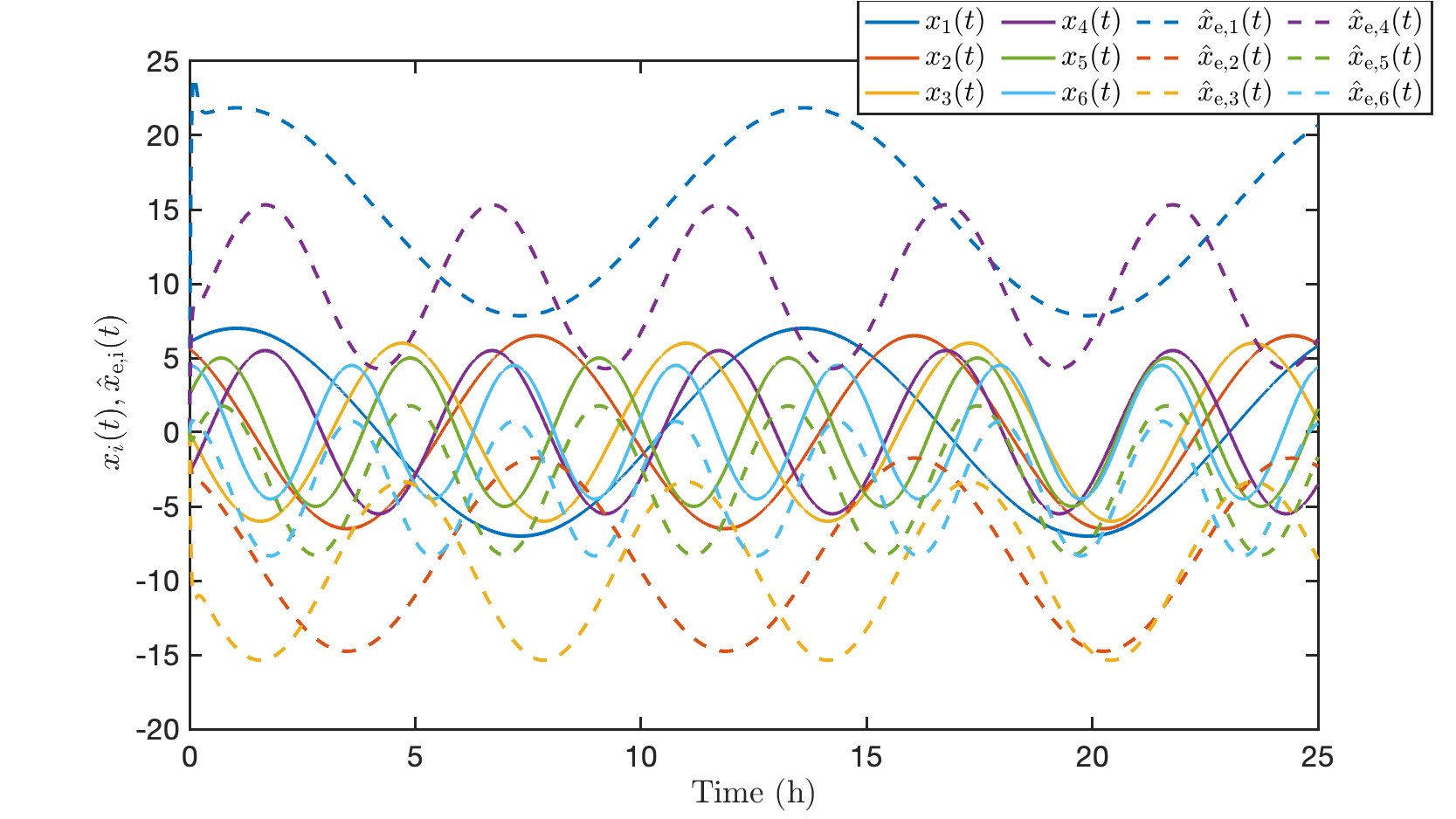}
\caption{Eavesdropper reconstruction of $x_i(t)$ ($\hat{x}_{\rm e,i}$) vs true reference signal $x_i(t)$ under the proposed algorithm. }
\label{fig:fig3}
\end{figure}

Next, we show that our algorithm preserves privacy against internal honest-but-curious agents. Consider agent~2 as an honest-but-curious agent attempting to infer agent~1’s reference \(x_1(t)\). Agent~6 is the other neighbor of agent~1 and behaves legitimately. By Theorem~\ref{hbc-privacy}, the honest-but-curious agent's coalition \(\mathcal H\) knows \(\eta_{21}\) and \(\eta_{12}\) but not \(\eta_{61}\) or \(\eta_{16}\), so
\[
m_1=\underbrace{(\eta_{21}-\eta_{12})}_{v_1\ \text{(known to }\mathcal H)}
\;+\;
\underbrace{(\eta_{61}-\eta_{16})}_{u_1\ \text{(unknown to }\mathcal H)}.
\]
Applying the attacker model in~\cite{zhang2022privacy}, agent~2 can obtain the masked signal
$
x_{{\rm m},1}(t)=x_1(t)+m_1=(x_1(t)+u_1)+v_1.$
With \(v_1\) known and \(u_1\) unknown, the best it can recover is
\[
\begin{aligned}
x_{{\rm m},1}(t)-v_1 &= x_1(t)+u_1
= x_1(t) +14.85- (6.75-0.2),\\
&= x_1(t) +8.3,
\end{aligned}
\]
which is offset by the unknown constant \(u_1=8.3\) and thus not the true reference \(x_1(t)\). Figure~\ref{fig:fig4} compares the honest-but-curious agent’s estimate \(\hat{x}_{\rm h,1}(t)\) with the true reference \(x_1(t)\), showing that the honest-but-curious agent cannot recover the true reference \(x_1(t)\).
\begin{figure}[!t]
\centering
\includegraphics[width=3.6in]{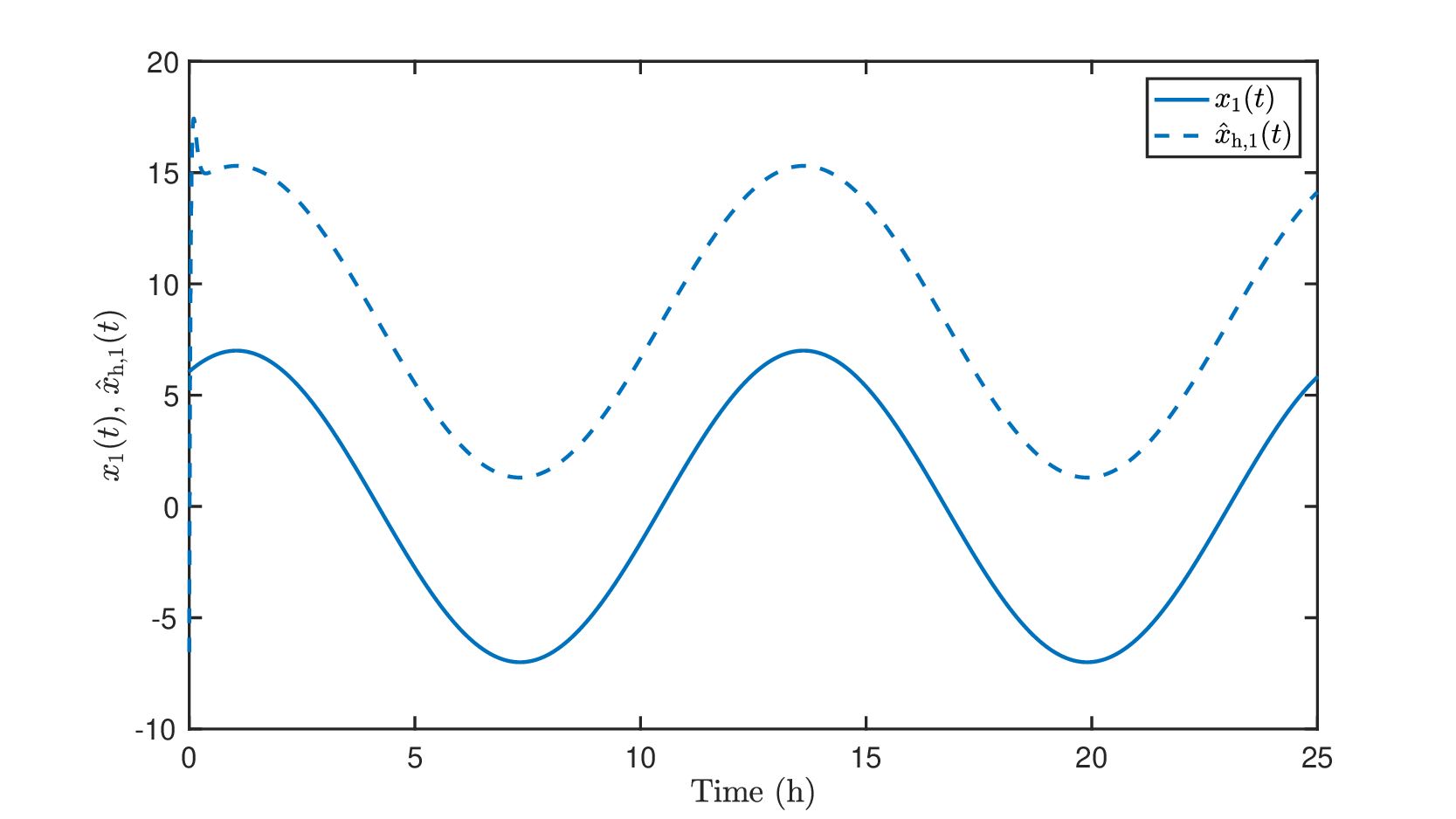}
\caption{An Honest-but-curious agent's (agent 2)  reconstruction of $x_1(t)$ ($\hat{x}_{h,1}(t)$) vs true reference signal $x_i(t)$ under the proposed algorithm. }
\label{fig:fig4}
\end{figure}

We now show that the proposed privacy-preserving algorithm preserves the convergence rate of the conventional algorithm and outperforms the state-decomposition method. Figure~\ref{fig:fig5} depicts the evolution of the consensus error $\|e(t)\|_2$, where $e_i(t)=\bigl|\hat{x}_{{\rm a},i}(t)-x_{\rm a}\bigr|$. The proposed algorithm matches the conventional algorithm’s rate while converging substantially faster than the state-decomposition baseline.
\begin{figure}[!t]
\centering
\includegraphics[width=3.6in]{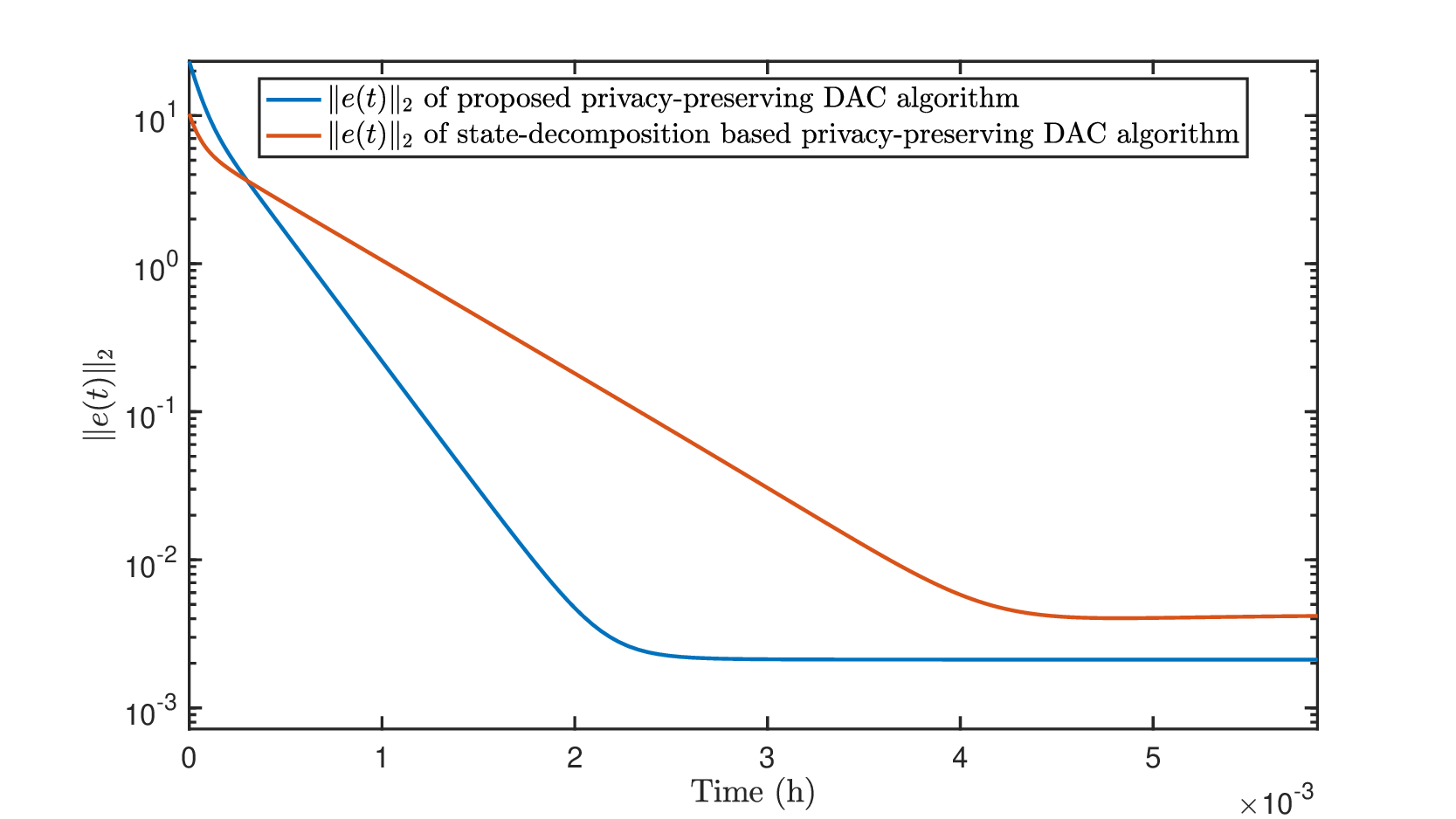}
\caption{$\ell_2$-norm of the consensus error for the proposed privacy-preserving DAC vs the state-decomposition-based DAC.}

\label{fig:fig5}
\end{figure}

\section{Conclusion}\label{sec:Conclusion}
This paper proposed a novel privacy-preserving DAC algorithm with masked reference signals. We proved that the algorithm's tracking error converges to the same neighborhood as the conventional dynamic average consensus algorithm with the same convergence rate while preserving privacy. We proved that the algorithm is able to protect each agent’s reference signal from disclosure to external eavesdroppers. We also showed that an honest-but-curious agent cannot infer a neighboring node’s reference signal when the target node has at least one legitimate neighbor. Finally, simulation results demonstrated the theoretical results.

\bibliographystyle{IEEEtran}
\bibliography{reference}

\end{document}